\documentclass[11pt]{article}
\usepackage{amsmath, amssymb, amsthm, amsfonts, amssymb, graphicx, epstopdf}
\usepackage{cite}
\usepackage{subfig}
\usepackage{array}
\usepackage{multirow}
\usepackage{longtable}
\usepackage{multicol}
\usepackage{cite}
\usepackage{float}
\usepackage{makecell}
\usepackage{wasysym,amsbsy}
\usepackage{scalerel}
\usepackage{xcolor}
\usepackage{geometry}
\usepackage{xspace}
\usepackage[mathscr]{euscript}
\usepackage{mathtools}
\usepackage{hyperref}
\usepackage[utf8]{inputenc}
\usepackage[english]{babel}
\usepackage{mathtools}
\usepackage{lscape}
\usepackage{mathrsfs}
\usepackage[scr=rsfs,cal=boondox]{mathalfa}
\usepackage[mathscr]{euscript}

\usepackage{xcolor, soul}
\sethlcolor{yellow}

\newtheorem{theorem}{\bf Theorem}[section]

\newtheorem{lemma}[]{\bf Lemma}[section]

\newtheorem{definition}[theorem]{\bf Definition}

\newcolumntype{P}[1]{>{\centering\arraybackslash}p{#1}}

\date{}

\addto\captionsenglish{}

\begin{document}
\title{\textbf{ HYPERGRAPH REPRESENTATION IN BRAIN NETWORK ANALYSIS}}
	\author{Anagha P, Selvakumar R$^{*}$ \\\\ Department of Mathematics, Vellore Institute of Technology, Vellore 632014, India\\\\
	Email: anuanagha01@gmail.com, rselvakumar@vit.ac.in$^{*}$}
\maketitle

\vspace{-0.5cm}
\begin{abstract}
For the study of functional aspects of the brain network. This paper is a study on the hypergraph representation, based on the functional regions of the brain network. A new parameter that can measure how many multifunctioning regions each function contains and thereby the correlation of other functions with each function. 

\end{abstract} 

\section{Introduction}

The human brain is the most intricately connected network ever discovered by mankind. The human brain is made up of approximately $10^{11}$ neurons that are connected by approximately $10^{14}$ synapses. In the light of graph theory, brain networks are made up of vertices (nodes) and edges, where vertices stand in for neurons or regions of the brain and edges stand in for the connections that are either structural or functional between vertices \cite{1,2}.\\
Studies on humans indicate that modular brain networks improve cognitive performance. The modularity of a network is a structural measure that evaluates how well the network can be partitioned into smaller sub-networks (also called groups, communities, or clusters ). As higher modularity reflects a larger number of intra-module connections and fewer inter-module connections, it is commonly believed that a highly modular brain consists of highly specialized brain networks with less integration across networks. Recent research on both younger and elderly individuals has demonstrated that preexisting differences in the modularity of brain networks can predict post-intervention performance improvements \cite{3,4}.\\
The first step in creating a brain network is defining the nodes and edges of the network.
The brain network edges show the connectivity between brain areas. The connectivity of the brain network can be classified as structural, functional, or effective connectivity.
Functional connectivity is a statistical association between brain regions and physiological or neurophysiological signals \cite{5,6}. Topological indices are important numerical quantities that reflect various connectivity properties of the brain network \cite{7,8}. There are several topological indices that are based on different things like eccentricity, degree, distance, and so on \cite{9,10,11}.
\\
The brain network can be modeled and analyzed using hypergraph representation.  Hypergraphs, compared to standard graphs, can represent more complex relationships between vertices than just connections or edges \cite{12}. Since hypergraphs are capable of reflecting complex relationships between nodes (brain regions), they can be used to model and analyse brain networks. The analysis of functional connectivity is a crucial use of hypergraphs in the study of brain networks \cite{12,13}. Functional connectivity describes the relationships between the levels of activity in various brain regions. By enabling numerous brain regions to be connected at once by a hyperedge, rather than just pairings of brain regions as in standard graphs, hypergraphs can aid in the capturing of complex functional relationships.\\
For example: Assume that $A$, $B$, and $C$ are neurons or brain regions, and that $A$, $B$, and $C$ share the same function.
If a standard graph were to depict this situation, only two of the three regions would have edges connecting them at once, resulting in a complete graph. But a hyperedge that represents the function includes all three in hypergraph representation.\\
Overall, hypergraphs provide a powerful tool for modeling and analyzing the intricate relationships between brain regions, allowing for a deeper understanding of neural activity and cognition.\\

\section{Hypergraph Topological Indices}

This section introduces a new parameter, hyperedge degree $d_h(\epsilon)$. It is a parameter that depends on the degree (connected to various functions) to which each vertex of this hyperedge. What is a region's involvement of different functions in the brain is more essential than what brain regions are connected to a function. Using this parameter, it is possible to determine which brain regions have an effect on brain function and to use this information for future brain research. \\
A brain network can be represented as a hypergraph with brain regions or neurons serving as vertices and brain functions as hyperedges. $d_h(\epsilon)$ will be high if certain brain areas or neurons involved in a given function $\epsilon$ involve more than one function or if there are more connections between $\epsilon$ and other hyperedges.\\
Novel topological indices based on hyoergraph degrees of some popular graphs are defined and discussed in this section. Also hypergraph degrees and new topological indices values for some family of graph with small-world organisation is studied. The fact that human brain networks prominently display small-world organisation is one of the most important results. This network architecture in the brain (the result of natural selection acting under the pressure of a cost-efficiency balance) enables the efficient segregation and integration of information with minimal wiring and energy costs.
Additionally, the small-world organisation experiences ongoing modifications as part of normal growth and ageing and shows significant changes in neurological and mental illnesses \cite{14}. \\
For the study's convenience, each hyperedge was treated as a complete graph and its $d_h(\epsilon)$ values were computed.\\

\begin{definition}
Let $d_h(v_i)$ is the number of $v$ contained hyperedges $\epsilon$ of $H$, then  $d_h(\epsilon)=\sum_{v_i\in \epsilon} d_h(v_i)-|\epsilon|$ is the hyperedge degree of $\epsilon$.
\end{definition}
\begin{definition}
The hyper first zagreb index and hyper first general zagreb index are defined as, $HFGZI(H)=\sum_{\forall \epsilon\in E} d_h(\epsilon)$ and $HM_1(H)=\sum_{\forall \epsilon\in E} d_h^2(\epsilon)$ respectively.
\end{definition}

\begin{lemma}
Here hyperedge is vertex subset $V(K_n)$, where $K_n$ represents complete graph. The hypergraph topological indices of some popular graphs are:
\begin{itemize}
    \item Let $K_n$ be a complete graph with $n$ vertices. Then $d_h(\epsilon)=(n-1)(n-2) $ $\forall \epsilon\in K_n$ and therefore $HFGZI(K_n)=n(n-1)(n-2)$ and $HM_1(K_n)=n(n-1)^2(n-2)^2$.
    \item Let $C_n$ be a cycle graph with $n$ vertices.Then $n(\epsilon)=n$ and $d_h(\epsilon)=2$ $ \forall \epsilon\in C_n$. So, $HFGZI(C_n)=2n$ and $HM_1 (C_n)=4n$.
    \item Let $T$ be a tree, then $d_h(\epsilon)=N(u)+N(v)-2 $ $\forall \epsilon\in T$, where $u,v \in \epsilon$ and $u\neq v$. So, $HFGZI(T)=\sum_{uv\in E(T)} (N(u)+N(v)-2) $ and $HM_1(T)=\sum_{uv\in E(T)} (N(u)+N(v)-2)^2$.\\
    In particular, 
         \begin{itemize}
\item Let $P_n$ be a path with $n$ vertices, then $d_h(\epsilon)=\begin{cases}
       1 & \text{;if $\epsilon$ is an end edge}\\
       2 & \text{; otherwise}\\
    \end{cases}$.\\ Therefore $HFGZI(P_n)=2(n-2)$ and $HM_1(P_n)= 2+4(n-3)$.
    \item $HFGZI(S_r)=r(r-1)$ and $HM_1(S_r)=r(r-1)^2$ where $S_r$ is a star graph of $r+1$ vertices and $d_h(\epsilon)=r-1$.
    \end{itemize}
\end{itemize}
\end{lemma}
\begin{proof}
In case of $K_n$, $K_{n-1}$ is the hyperedge. In case of $C_n$ and tree $T$, each edge $K_2$ is the hyperedge. So, the result is obvious.
   
\end{proof}

The structural and functional networks of the human brain are organized in a small-world structure. The small-world model quantifies the separation and integration of information. Individual cognition is captured by the small-world paradigm, which also has a physiological basis. So now the new parameter value and indices for the graph with small-world organisation are going to be discussed here. This section simplifies calculation by treating complete subgraphs as hyperedges.\\

\begin{lemma}
Let $G \cong W_p^q $ (Windmill graph) then number of hyperedges in $G$, $n(E)=q$ and
$d_h (\epsilon)= q-1 $ $\forall \epsilon\in G$.
\end{lemma}
\begin{proof}
The total number of hyperedges in $W_p^q$ is $q$. So,\\
$d_h(v)=\begin{cases}
       q & \text{;if $v$ is the center}\\
       1 & \text{; otherwise}\\
    \end{cases}$ and hence\\
$ d_h(\epsilon)=\sum_{v\in \epsilon} d_h(v)-|\epsilon|=q-1$
\end{proof}
\begin{theorem}
If $G$ be $W_p^q$ then $HFGZI(G)=q(q-1)$ and $HM_1 (G)=q(q-1)^2$.
\end{theorem}
\begin{proof}
Result is obvious from lemma(2.2)
\end{proof} 

\begin{lemma}
Let $G \cong F_{r,s}$ then number of hyperedges in $G$, $n(E)=r+s$ and
$d_h (\epsilon)= r+s-1 $ $\forall \epsilon\in G$
\end{lemma}
\begin{proof}
Since $F_{r,s}$ contains $r$ triangles (means $K_3$) and s pendent edges (means $K_2$) and each complete graph is an hyperedge, the total number of hyperedges is $r+s$. So,\\
$d_h(v)=\begin{cases}
       r+s & \text{;if $v$ is the center}\\
       1 & \text{; otherwise}\\
    \end{cases}$ and hence \\
    $\begin{array}{ll}   d_h(\epsilon)&=\sum_{v\in \epsilon} d_h(v)-|\epsilon|\\
&=\begin{cases}
       1+1+(r+s)-3 & \text{;if $\epsilon$ is $K_3$}\\
       1+r+s-2 & \text{; if $\epsilon$ is $K_2$}\\
    \end{cases}\\
    &=\begin{cases}
       r+s-1 & \text{;if $\epsilon$ is $K_3$}\\
       r+s-1 & \text{; if $\epsilon$ is $K_2$}\\
    \end{cases}\\
\end{array}$\\
\end{proof}
\begin{theorem}
If $G$ be $F_{r,s}$ then $HFGZI(G)=(r+s)(r+s-1)$ and $HM_1 (G)=(r+s)(r+s-1)^2$.
\end{theorem}
\begin{proof}
Result is obvious from lemma(2.3)
\end{proof} 

\begin{lemma}
Let $G \cong W_n$ (Wheel graph with $n$ vertices) then number of hyperedges in $G$, $n(E)=n$ and $d_h (\epsilon)= n+1 $ $\forall \epsilon\in G$.
\end{lemma}
\begin{proof}
Since wheel graph contains $n$ triangles (means $K_3$) and each complete graph is an hyperedge, the total number of hyperedges is $n$. So,\\
$d_h(v)=\begin{cases}
       n & \text{;if $v$ is the center}\\
       2 & \text{; otherwise}\\
    \end{cases}$ and hence\\
$d_h(\epsilon)=\sum_{v\in \epsilon} d_h(v)-|\epsilon|
=2+2+n-3
=n+1$
\end{proof}
\begin{theorem}
If $G$ be a Wheel graph $W_n$ then $HFGZI(G)=n(n+1)$ and $HM_1 (G)=n(n+1)^2$.
\end{theorem}
\begin{proof}
Result is obvious from lemma(2.4)
\end{proof} 
 
\section{Graph Operations}
To construct a large network from small networks and viceversa, graph operations are helpful. Graph operations join, cartesian product, corona products and composition are defined as, the cartesian product $G_1\times G_2$ of graphs $G_1$ and $G_2$ is a graph with vertex set $V(G_1\times G_2)=V(G_1)\times V(G_2)$ and $(a,x)(b,y)$ is an edge of $G_1\times G_2$ if $a=b$ and $xy \in G_2$, or $ab \in E(G_1)$ and $x=y$;
the join $G_1+G_2$ of graphs $G_1$ and $G_2$ is a graph with vertex set $V(G_1)\cup V(G_2)$ and edge set $E(G_1)\cup E(G_2)\cup \{uv; u\in V(G_1)$ and $v\in V(G_2)\} $;
the composition $G_1\circ G_2$ of graphs $G_1$ and $G_2$ with disjoint vertex sets $V(G_1)$ and $V(G_2)$ and edge sets $E(G_1)$ and $E(G_2)$ is the graph with vertex set $V(G_1)\times V(G_2)$ and $u=(u_1,v_1)$ is adjacent to $v=(u_2,v_2)$ whenever $u_1$ is adjacent to $u_2$ or $u_1=u_2$ and $v_1$ is adjacent to $v_2$; 
The corona product $G_1\odot G_2$ is defined as the graph obtained from $G_1$ and $G_2$ by taking one copy of $G_1$ and $|V(G_1)|$ copies of $G_2$ and then joining by an edge each vertex of the $i$th copy of $G_2$ is named $(G_2,i)$ with the $i$th vertex of $G_1$ \cite{18,19}.\\

Cartesian product of any two complete graphs $G_1$ and $G_2$ results in a graph with hyperedges collection of $G_1$ and $G_2$. 
\begin{lemma}
    Let $G_1=K_n$ and $G_2=K_m$ then cartesian product $G=G_1\times G_2$ of hypergraphs $G_1$ and $G_2$ is a hypergraph with vertex set $V(G)=V(G_1) \times V(G_2)$ and edge set \\ $E(G)=\{E(G_1)(m $ times$),E(G_2)(n$ times$)\}$.
\end{lemma}
\begin{proof}
    From definition of hypergraph and cartesian product of graphs
\end{proof}
\begin{theorem}
     Let $G=G_1\times G_2$ be cartesian product of hypergraphs where $G_1=K_n$ and $G_2=K_m$ then $G$ contains $n+m$ hyperedges and 
     $d_h(\epsilon)=\begin{cases}
       n & \text{;if $\epsilon$ is $K_n$}\\
       m & \text{;if $\epsilon$ is $K_m$}\\
    \end{cases}$  and $HFGZI(G)=2|V(G_1)||V(G_2)|$ and $HM_1(G)=|V(G_1)||V(G_2)|(|V(G_1)|+|V(G_2)|)$
\end{theorem}
\begin{proof}
From lemma(3.1), clear that $n(E)=n+m$.
Here  $E(G)=\{K_n,...,K_n(m $ times$), K_m,...,K_m(n $ times$)\}$, $d_h(v)=2 $ $\forall v \in G$ and $d_h(\epsilon)=\sum_{v\in \epsilon} d_h(v)-|\epsilon|$. Therefore
$d_h(K_n)= 2+2+...+2(n$ times$)-n=2n-n=n$ and $d_h(K_m)= 2+2+...+2(m $ times$)-m=2m-m=m$. So,\\
 $d_h(\epsilon)=\begin{cases}
       n & \text{;if $\epsilon$ is $K_n$}\\
       m & \text{;if $\epsilon $is $K_m$}\\
    \end{cases}$
    \\
    $\begin{array}{ll}HFGZI(G)&=\sum_{\forall \epsilon \in G_1\times G_2} d_h(\epsilon)\\
    &=\sum_{\forall K_n} d_h(\epsilon)+\sum_{\forall K_m} d_h(\epsilon)\\
    &= m(2n-n)+n(2m-m)\\
    &=2nm\\
    &=2|V(G_1)||V(G_2)|\\
    \end{array}$\\

    $\begin{array}{ll} HM_1(G)&=\sum_{\forall \epsilon \in G_1\times G_2} d_h^2(\epsilon)\\
    &=\sum_{\forall K_n} d_h^2(\epsilon)+\sum_{\forall K_m} d_h^2(\epsilon)\\
    &= m(2n-n)^2+n(2m-m)^2\\
    &=nm(n+m)\\
    &=|V(G_1)||V(G_2)|(|V(G_1)|+|V(G_2)|)\\
    \end{array}$\\
\end{proof}

\begin{lemma}
Join product $G=G_1+G_2$ of hypergraphs $G_1$ and $G_2$ is a hypergraph with vertex set $V(G)=V(G_1)\cup V(G_2)$ and edge set $E(G)=\{\epsilon+ \epsilon^*; \forall \epsilon \in E(G_1) $ and $ \epsilon^*\in E(G_2)\} $.
\end{lemma}

\begin{theorem}
Let the hypergraph $G=G_1+G_2$ be join of $G_1$ and $G_2$ and $\epsilon = \epsilon'+\epsilon^*$ be a hyperedge of $G_1 + G_2$, then $G$ contains $n_1n_2$ hyperedges where $n_1$ is 
the number of hyperedges in $G_1$ and $ n_2$ is the number of hyperedges in $G_2$ and $d_h(\epsilon)=n_2(d_h(\epsilon')+|\epsilon'|)+n_1(d_h(\epsilon^*+|\epsilon^*|)$, where $ \epsilon' \in E(G_1) $ and $ \epsilon^* \in E(G_2)$ and $HFGZI(G)=n_2^2 HFGZI(G_1)+n_1^2 HFGZI(G_2)+n_2(n_2-1) \sum_{\forall \epsilon'} |\epsilon'|+n_1(n_1-1) \sum_{\forall \epsilon^*} |\epsilon^*|$.

\end{theorem}
\begin{proof}
    Let $G_1$ contains $n_1$ hyperedges and $G_2$ contains $n_2$ hyperedges then number of hyperedges in G, 
    $n(E(G))=n(E(G_1+G_2))=n(E(G_1))\times n(E(G_2))=n_1n_2$ and  
    $d_h(V)=\begin{cases}
    n_2 d_h(v) & \text{; if $v\in V(G_1)$}\\
    n_1 d_h(v) & \text{; if $v\in V(G_2)$}\\
    \end{cases}$. \\
    Let $\epsilon'_1,\epsilon'_2,...,\epsilon'_{n_1}$ are hyperedges of $G_1$ and $\epsilon^*_1,\epsilon^*_2,...,\epsilon^*_{n_2}$ are hyperedges of $G_2$, then\\ $E(G)=E(G_1+G_2)=\{(\epsilon'_1+\epsilon_1^*),(\epsilon'_1+\epsilon_2^*),...,(\epsilon'_1+\epsilon_{n_2}^*), (\epsilon'_2+\epsilon_1^*),(\epsilon'_2+\epsilon_2^*),...,(\epsilon'_2+\epsilon_{n_2}^*),..., (\epsilon'_{n_1}+\epsilon_1^*),(\epsilon'_{n_1}+\epsilon_2^*),...,(\epsilon'_{n_1}+\epsilon_{n_2}^*)$.
    Let $\epsilon' \in E(G_1)$ and $\epsilon^*\in E(G_2)$ then\\
 $\begin{array}{ll}  d_{h_{G_1+G_2}}(\epsilon)&=d_h(\epsilon'+\epsilon^*);  \epsilon' \in G_1,  \epsilon^*\in G_2\\
 &=\sum_{V\in V(\epsilon'+\epsilon^*)} d_h(V)-|\epsilon'+\epsilon^*|\\
 &=n_2 \sum_{v\in V(\epsilon')} d_h (v)+n_1 \sum_{v^*\in V(\epsilon^*)} d_h(v^*)-|\epsilon'|-|\epsilon^*|\\
 &=n_2 d_h(\epsilon')+n_1 d_h(\epsilon^*)+(n_2-1)|\epsilon'|+(n_1-1)|\epsilon^*|
 \end{array}$\\
  $\begin{array}{ll} HFGZI(G_1+G_2)&=\sum_{\epsilon\in E(G_1+G_2)} d_h(\epsilon)\\
  &=\sum_{\forall \epsilon' \in E(G_1), \epsilon^*\in E(G_2)} d_h(\epsilon'+\epsilon^*)\\
  &= n_2 (d_h(\epsilon'_1)+|\epsilon'_1|)+n_1 (d_h(\epsilon^*_1)+|\epsilon^*_1|-(|\epsilon'_1|+|\epsilon^*_1|))+n_2 (d_h(\epsilon'_1)+|\epsilon'_1|)\\& +n_1 (d_h(\epsilon^*_2)+|\epsilon^*_2|)-(|\epsilon'_1|+|\epsilon^*_2|)+...+n_2 (d_h(\epsilon'_1)+|\epsilon'_1|)+n_1 (d_h(\epsilon^*_{n_2})+|\epsilon^*_{n_2}|)\\& -(|\epsilon'_1|+|\epsilon^*_{n_2}|)+n_2 (d_h(\epsilon'_2)+|\epsilon'_2|)+n_1 (d_h(\epsilon^*_1)+|\epsilon^*_1|)-(|\epsilon'_2|+|\epsilon^*_1|)\\& +n_2 (d_h(\epsilon'_2)+|\epsilon'_2|)+n_1 (d_h(\epsilon^*_2)+|\epsilon^*_2|)-(|\epsilon'_2|+|\epsilon^*_2|)+...+n_2 (d_h(\epsilon'_2)+|\epsilon'_2|)\\& +n_1 (d_h(\epsilon^*_{n_2})+|\epsilon^*_{n_2}|)-(|\epsilon'_2|+|\epsilon^*_{n_2}|)+...+n_2 (d_h(\epsilon'_{n_1})+|\epsilon'_{n_1}|)+n_1 (d_h(\epsilon^*_1)+|\epsilon^*_1|)\\& -(|\epsilon'_{n_1}|+|\epsilon^*_1|)+n_2 (d_h(\epsilon'_{n_1})+|\epsilon'_{n_1}|)+n_1 (d_h(\epsilon^*_2)+|\epsilon^*_2|)-(|\epsilon'_{n_1}|+|\epsilon^*_2|)\\& +...+n_2 (d_h(\epsilon'_{n_1})+|\epsilon'_{n_1}|)+n_1 (d_h(\epsilon^*_{n_2})+|\epsilon^*_{n_2}|)-(|\epsilon'_{n_1}|+|\epsilon^*_{n_2}|)\\
  &=n_2^2 \sum_{i=1}^{n_1}( d_h(\epsilon'_i)+|\epsilon'_i|)+n_1^2 \sum_{j=1}^{n_2}( d_h(\epsilon^*_j)+|\epsilon^*_j|)-(n_2\sum_{i=1}^{n_1}|\epsilon'_i|+ n_1\sum_{j=1}^{n_2}|\epsilon^*_j|)\\
  &=n_2^2 \sum_{\forall \epsilon' \in E(G_1)}( d_h(\epsilon')+|\epsilon'|)+n_1^2 \sum_{\forall \epsilon^* \in E(G_2)}( d_h(\epsilon^*)+|\epsilon^*|)\\& -(n_2\sum_{\forall \epsilon' \in E(G_1)}|\epsilon'| + n_1\sum_{\forall \epsilon^* \in E(G_2)}|\epsilon^*|)\\
  &=n_2^2 HFGZI(G_1)+n_1^2 HFGZI(G_2)+n_2(n_2-1) \sum_{\forall \epsilon'} |\epsilon'|+n_1(n_1-1) \sum_{\forall \epsilon^*} |\epsilon^*|\\
  \end{array}$\\
\end{proof}

\begin{lemma}
Let $G_1=S_r$ and $G_2=K_n$ then corona product $G=G_1\odot G_2$ of $G_1$ and $G_2$ is a hypergraph with edge set $E(G)=\{K_{n+1}((r+1)$ times$),K_2(r$ times$)\}$ and $|V(G)|=(n+1)(r+1)$.
\end{lemma}

\begin{theorem}
Let $G=G_1\odot G_2$ be corona product of $G_1=S_r$ and $G_2=K_n$. Then $G$ contains $2r+1$ hyperedges and $d_h(\epsilon)=\begin{cases}
       r+1 & \text{;if $\epsilon$ is $K_2$ (the pendent edge of $S_r$)}\\
       r & \text{;if $\epsilon$ is the $K_{n+1}$ attached to the center}\\
       1 & \text{;otherwise}\\
    \end{cases}$  and $HFGZI(G)=HFGZI(G_1)+4r$ and $HM_1(G)=HM_1(G_1)+5r^2+r$
\end{theorem}
\begin{proof}
From lemma(3.4), clear that $n(E)=2r+1$.
Here  $E(G)=\{K_{n+1}((r+1)$ times$),K_2(r$ times$)\}$,
$d_h(v)=\begin{cases}
       r+1 & \text{;if $v$ is the center of $S_r$}\\
       2 & \text{;if $v$ is the pendent vertex of $S_r$}\\
       1 & \text{;otherwise}\\
 \end{cases}$\\
 and $d_h(\epsilon)=\sum_{v\in \epsilon} d_h(v)-|\epsilon|$. Therefore
$d_h(K_2)= 2+r+1-2=r+1$, $d_h(K_{n+1};$one attached to the center$)= (r+1)+1+1+...+1(n$ times$)-(n+1)=r$ and $d_h(K_{n+1};$except one attached to the center$)= 2+1+1+...+1(n$ times$)-(n+1)=1$. So,\\
  $d_h(\epsilon)=\begin{cases}
       r+1 & \text{;if $\epsilon$ is $K_2$ (the pendent edge of $S_r$)}\\
       r & \text{;if $\epsilon$ is the $K_{n+1}$ attached to the center}\\
       1 & \text{;otherwise}
 \end{cases}$\\
 $HFGZI(G)=\sum_{\epsilon \in G} d_h(\epsilon)
    =r\times(r+1)+1\times r+r \times 1 
    =r^2+3r=r(r-1)+4r=HFGZI(G_1)+4r$\\
 $HM_1(G)=\sum_{\epsilon \in G} d_h^2(\epsilon)
    =r\times(r+1)^2+1\times r^2+r \times 1^2 
    =r(r-1)^2+5r^2+r=HM_1(G_1)+5r^2+r$\\
\end{proof}

\begin{lemma}
Let $G_1=K_n$ and $G_2=K_m$ then corona product $G=G_1\odot G_2$ of hypergraphs $G_1$ and $G_2$ is a hypergraph with edge set $E(G)=\{K_{m+1}(n$ times$),K_n\}$ and $|V(G)|=n(m+1)$.
\end{lemma}

\begin{theorem}
Let $G=G_1\odot G_2$ be corona product of $G_1=K_n$ and $G_2=K_m$. Then $G$ contains $n+1$ hyperedges and $d_h(\epsilon)=\begin{cases}
       1 & \text{;if $\epsilon$ is $K_{m+1}$}\\
       n & \text{;if $\epsilon$ is $K_n$}\\
    \end{cases}$  and $HFGZI(G)=2n(G_1)$ and $HM_1(G)=n(G_1)[n(G_1)+1]$
\end{theorem}
\begin{proof}
From lemma(3.5), clear that $n(E)=n+1$.
Here  $E(G)=\{K_{m+1}(n$ times$),K_n\}$,\\
$d_h(v)=\begin{cases}
       2 & \text{;if $v \in V(K_n)$}\\
       1 & \text{;if $v \in V(K_m)$}\\
 \end{cases}$ and $d_h(\epsilon)=\sum_{v\in \epsilon} d_h(v)-|\epsilon|$. Therefore
$d_h(K_n)= 2+2+...+2(n$ times$)-n=n$ and $d_h(K_{m+1})=1$. So,\\
  $d_h(\epsilon)=\begin{cases}
       1 & \text{;if $\epsilon$ is $K_{m+1}$}\\
       n & \text{;if $\epsilon$ is $K_n$}\\
    \end{cases}$ \\
 $HFGZI(G)=\sum_{\epsilon \in G} d_h(\epsilon)
    =n\times1+1\times n=2n=2n(G_1)$\\
 $HM_1(G)=\sum_{\epsilon \in G} d_h^2(\epsilon)
    =n\times1^2+1\times n^2=n+n^2=n(G_1)[n(G_1)+1]$
\end{proof}

\section{Conclusion}
The brain is the primary organ that regulates all body functions. Numerous functions are controlled by the brain. Each function is regulated by multiple regions, and each region contains multiple functions. In this context, hypergraphs are more useful than standard or conventional graphs. Normal graphs only indicate whether neurons or brain regions are functionally connected or not; it is uncertain which function links these neurons.\\ 
So that this study represented the brain as a hypergraph (brain regions as nodes, and each function as a hyperedge). $d_h(\epsilon)$ value indicate the intensity of interconnections therefore introduction of these parameter is useful in brain network analysis.

\
\end{document}